\documentclass[conference,10pt,letterpaper]{IEEEtran}

\usepackage[letterpaper, top=0.5in, bottom=0.58in, left=0.75in, right=0.75in]{geometry}

\usepackage[english]{babel}
\usepackage{amsthm}
\usepackage{enumitem}
\usepackage{amsmath}
\usepackage[english]{babel}
\usepackage{amssymb}
\usepackage{mathtools}

\usepackage{enumerate}
\usepackage{dutchcal}
\usepackage{tabularx}
\usepackage{mathptmx}
\DeclareMathAlphabet{\mathcal}{OMS}{cmsy}{m}{n}  
\usepackage{cite} 
\usepackage{graphicx} 
\usepackage{float} %
\usepackage{setspace}
\usepackage[justification=centering]{caption}
	
\usepackage{authblk}
\usepackage[table]{xcolor}
\usepackage{times}
\usepackage{epsfig}
\usepackage{amsfonts}
\usepackage{graphicx}
\usepackage{amstext}
\usepackage{latexsym}
\usepackage{color,colortbl}
\usepackage{ifthen}
\usepackage{multirow}
\usepackage{verbatim}
\usepackage{array,tabularx}
\usepackage{arydshln}
\usepackage[mathscr]{euscript}
\usepackage{accents}
 \usepackage{cite}
\usepackage{hhline}
\usepackage[singlelinecheck=false,justification=justified]{caption}
\usepackage{subcaption}
\usepackage{enumerate}
\usepackage{mathtools}
\usepackage{url}
\usepackage{xparse}
\usepackage{makecell}
\usepackage{varwidth}
\usepackage{bm}
\usepackage{arydshln}
\usepackage{comment}
\usepackage{wrapfig}
\usepackage{float}
\usepackage{booktabs}
\usepackage{dirtytalk}
\usepackage{mathtools}
\usepackage{booktabs}
\usepackage{float}
\usepackage{multirow}
\usepackage{balance}
\usepackage[ruled,vlined]{algorithm2e}

\newcolumntype{C}[1]{>{\centering\let\newline\\\arraybackslash\hspace{0pt}}m{#1}}

\newcommand{\dist}{\operatorname{dist}} 
 
\usepackage[normalem]{ulem}
\usepackage{amsmath,pgfplots,amssymb}
\usepackage{graphicx}
\newtheorem{theorem}{Theorem}

\newtheorem{corollary}{Corollary}
\theoremstyle{definition}
\newtheorem{example}{Example}
\newtheorem{definition}{Definition}
\theoremstyle{definition}

\theoremstyle{definition}

\newcommand{\interior}[1]{%
  {\kern0pt#1}^{\mathrm{o}}%
}

\DeclareMathOperator{\sym}{Sym}

\usetikzlibrary{matrix,decorations.pathreplacing}

\newcommand{\restrictedto}[1]{|_{#1}}

\definecolor{myblue1}{rgb}{0.01, 0.28, 1.0}
\definecolor{myblue2}{rgb}{0.1, 0.1, 0.44}
\definecolor{myblue3}{rgb}{0.15, 0.23, 0.89}
\definecolor{myblue4}{rgb}{0.0, 0.22, 0.66}
\definecolor{myblue5}{rgb}{0.06, 0.75, 0.99}

\definecolor{myred1}{rgb}{0.89, 0.02, 0.17}
\definecolor{myred2}{rgb}{0.9, 0.13, 0.13}

\definecolor{mygreen1}{rgb}{0.0, 0.65, 0.31}
\definecolor{mygreen2}{rgb}{0.3, 0.73, 0.09}
\definecolor{mygreen3}{rgb}{0.2, 0.9, 0.2}

\makeatletter
\renewcommand{\fnum@figure}{Fig. \thefigure}
\makeatother

\usepackage{environ}
\NewEnviron{killcontents}{}

\usepackage{nidanfloat}

\usepackage[utf8]{inputenc} 
\usepackage[T1]{fontenc}
\usepackage{url}
\usepackage{ifthen}
\usepackage{cite}

\usepackage{tikz}
\usetikzlibrary{positioning} 
\usepackage{xcolor}

\begin{document}
\title{Rainbow Differential Privacy} 

\author{Ziqi Zhou\IEEEauthorrefmark{1}, Onur G{\"u}nl{\"u}\IEEEauthorrefmark{2}, Rafael G. L. D'Oliveira\IEEEauthorrefmark{3}, Muriel Médard\IEEEauthorrefmark{4}, Parastoo Sadeghi\IEEEauthorrefmark{5}, and Rafael F. Schaefer\IEEEauthorrefmark{2}\\ \IEEEauthorrefmark{1}Department of Electrical Engineering and Computer Science, Technical University of Berlin, Germany\\ \IEEEauthorrefmark{2}Chair of Communications Engineering and Security, University of Siegen, Germany\\\IEEEauthorrefmark{3}SMSS, Clemson University, USA, \qquad \IEEEauthorrefmark{4}RLE, Massachusetts Institute of Technology, USA\\ \IEEEauthorrefmark{5}SEIT, University of New South Wales, Canberra, Australia \\  ziqi.zhou@campus.tu-berlin.de, \{onur.guenlue, rafael.schaefer\}@uni-siegen.de,\\ rdolive@clemson.edu, medard@mit.edu, p.sadeghi@unsw.edu.au}

\maketitle

\begin{abstract} 
We extend a previous framework for designing differentially private (DP) mechanisms via randomized graph colorings that was restricted to binary functions, corresponding to colorings in a graph, to multi-valued functions. As before, datasets are nodes in the graph and any two neighboring datasets are connected by an edge. In our setting, we assume that each dataset has a preferential ordering for the possible outputs of the mechanism, each of which we refer to as a rainbow. Different rainbows partition the graph of datasets into different regions. We show that if the DP mechanism is pre-specified at the boundary of such regions and behaves identically for all same-rainbow boundary datasets, at most one optimal such mechanism can exist and the problem can be solved by means of a morphism to a line graph. We then show closed form expressions for the line graph in the case of ternary functions. Treatment of ternary queries in this paper displays enough richness to be extended to higher-dimensional query spaces with preferential query ordering, but the optimality proof does not seem to follow directly from the ternary proof.
\end{abstract}

\section{Introduction}
Differential privacy (DP), proposed in \cite{dwork2006calibrating, dworkoriginal2006}, is a general privacy-preserving framework that aims to limit the statistical capability of a curious analyst, regardless of its computational power, in determining whether or not the data of a specific individual was used in response to its query\footnote{DP variants assuming a finite compute power for the adversary have been studied in works including \cite{computationalDP} but are not within the scope of our work.}. Since its inception, DP has attracted extensive research effort; see \cite{survey_2017} for a survey and \cite{DworkBook} for a treatment of the subject. Recent high-profile applications of DP include the  2020 US Census \cite{uscensus}, as well as by Google, Apple and Microsoft~\cite{google,apple,microsoft}.

The DP constraints are defined on any neighboring datasets that differ in data from one individual. These constraints are \emph{local, relative and dataset-independent} contributing to the success of DP as a privacy preserving framework. However, an undesirable byproduct has been that many DP implementations are agnostic to the actual dataset at hand. Indeed, a vast majority of output perturbation DP mechanisms  take the worst-case query sensitivity between any two neighboring datasets to determine the scale of noise \cite{DworkBook}. This is a pessimistic approach and can adversely affect query utility \cite{individual}. Several fixes have been proposed to improve utility. In one direction, noise calibration to smooth sensitivity was proposed in \cite{nissim2007smooth}, for which a chosen utility level is not guaranteed and the mechanism suffers from a heavy tail leading to outliers. Another direction is relaxation of the DP constraints \cite{individual, blowfish, profile}. For example, \cite{individual} proposed \emph{individual}-DP, which defines DP constraints only between the \emph{given realization} of a dataset and its neighbors. This, however, destroys group DP, i.e., implied DP constraints between non-neighboring datasets no longer remain. Recently, \cite{RafnoDP2colorPaper} proposed designing dataset-dependent DP mechanisms for binary-valued queries that guarantee optimal utility and yet, do not weaken the original DP constraints in any way; see also \cite{BinaryOutputDP}. Each dataset has a true query value (e.g., blue or red) and is represented as a node on a graph with edges representing neighboring datasets. Let mechanism randomness be homogeneously pre-specified only at the boundary datasets. \cite{RafnoDP2colorPaper} showed how these initial constraints can be optimally extended in closed-form for all other datasets, where the probability of giving the truthful query response is maximized by taking into account the distance to the boundary while tightly satisfying all $(\varepsilon, \delta)$-DP constraints.

In this paper, we consider a strict extension of \cite{RafnoDP2colorPaper} by increasing the number of possible query outputs (e.g., \texttt{blue}, \texttt{red}, or \texttt{green} representing majority votes among three choices). This extension is challenging in several ways. In the binary case, optimal probability assignment to outputting one value (e.g., \texttt{blue}) automatically specifies the whole mechanism. In the multivariate case this is not enough. In order to circumvent this, we assume a preferential order, which resembles a rainbow, in outputting query values and solve the problem sequentially (e.g., first for \texttt{blue} then \texttt{red}, where \texttt{green} is automatically specified if we consider three colors). Second, solving the problem optimally to achieve approximate $(\varepsilon, \delta)$-DP appears hard for the multivariate case, so we consider $\delta = 0$ (i.e., pure DP). With these simplifications, we are then able to provide the following results. 

If all boundary datasets have their mechanism pre-specified such that the mechanism is identical for all same-preference boundary datasets, then at most one unique such mechanism with order reasonable optimal utility can be found. Furthermore, the problem can then be reduced to a line (path) graph. We present a closed form solution to line graphs for the case of three colors. This ternary solution recovers the binary case of \cite{RafnoDP2colorPaper} as a special case, but a new general pattern emerges: the first preferred value has up to two operating regimes that are characterized by its boundary probability and $\varepsilon$. However, the second preferred value can exhibit up to three operating regimes that additionally depend on the sum of the boundary probabilities of the first and second highest-priority values. Treatment of ternary queries in this paper displays enough richness to be extended to higher-dimensional query spaces with preferential query ordering, but the optimality proof does not seem to follow directly from the ternary proof.

\usetikzlibrary{positioning}
\tikzset{w/.pic={
    \clip (0,0) circle (0.25);
    \begin{scope}
        \fill[white] (-0.25,0) rectangle (0.25,0.25);
        \fill[white] (-0.25,-0.25) rectangle (0.25,0);
    \end{scope}
    \draw (0,0) circle (0.25);
    \path [draw=white,line width=5pt](-0.25,0)--(0.25,0);}}

\usetikzlibrary{positioning}
\tikzset{brg/.pic={
    \clip (0,0) circle (0.25);
    \begin{scope}
        \fill[blue] (-0.25,0) rectangle (0.25,0.25);
        \fill[green] (-0.25,-0.25) rectangle (0.25,0);
    \end{scope}
    \draw (0,0) circle (0.25);
    \path [draw=red,line width=5pt](-0.25,0)--(0.25,0);}}

\tikzset{bgr/.pic={
    \clip (0,0) circle (0.25);
    \begin{scope}
        \fill[blue] (-0.25,0) rectangle (0.25,0.25);
        \fill[red] (-0.25,-0.25) rectangle (0.25,0);
    \end{scope}
    \draw (0,0) circle (0.25);
    \path [draw=green,line width=5pt]
(-0.25,0) -- (0.25,0);}}

\tikzset{rgb/.pic={
    \clip (0,0) circle (0.25);
    \begin{scope}
        \fill[red] (-0.25,0) rectangle (0.25,0.25);
        \fill[blue] (-0.25,-0.25) rectangle (0.25,0);
    \end{scope}
    \draw (0,0) circle (0.25);
    \path [draw=green,line width=5pt]
(-0.25,0) -- (0.25,0);}}

\tikzset{rbg/.pic={
    \clip (0,0) circle (0.25);
    \begin{scope}
        \fill[red] (-0.25,0) rectangle (0.25,0.25);
        \fill[green] (-0.25,-0.25) rectangle (0.25,0);
    \end{scope}
    \draw (0,0) circle (0.25);
    \path [draw=blue,line width=5pt]
(-0.25,0) -- (0.25,0);}}

\tikzset{gbr/.pic={
    \clip (0,0) circle (0.25);
    \begin{scope}
        \fill[green] (-0.25,0) rectangle (0.25,0.25);
        \fill[red] (-0.25,-0.25) rectangle (0.25,0);
    \end{scope}
    \draw (0,0) circle (0.25);
    \path [draw=blue,line width=5pt]
(-0.25,0) -- (0.25,0);}}

\tikzset{grb/.pic={
    \clip (0,0) circle (0.25);
    \begin{scope}
        \fill[green] (-0.25,0) rectangle (0.25,0.25);
        \fill[blue] (-.25,-0.25) rectangle (0.25,0);
    \end{scope}
    \draw (0,0) circle (0.25);
    \path [draw=red,line width=5pt] (-0.25,0) -- (0.25,0);}}

\begin{figure*}[!t]
\centering
\begin{subfigure}[t]{0.45\textwidth}
    \centering
    \begin{tikzpicture}
    \path pic {brg} node [align=center, above=2mm, scale=0.7, black] {a};
    \path pic [xshift=1cm] {brg} node [align=center, xshift=1cm, above=2mm, scale=0.7, black] {b};
    \path pic [xshift=2cm] {brg} node [align=center, xshift=2cm, above=2mm, scale=0.7, black] {c};
    \path pic [xshift=3cm] {rgb} node [align=center, xshift=3cm, above=2mm, scale=0.7, black] {d};
    \path pic [xshift=4cm] {rgb} node [align=center, xshift=4cm, above=2mm, scale=0.7, black] {e};
    \path pic [xshift=0.5cm, yshift=-1cm] {brg} node [align=center, xshift=0.51cm, yshift=-1cm, above=2mm, scale=0.7, black] {f};
    \path pic [xshift=1.5cm, yshift=-1cm] {brg} node [align=center, xshift=1.5cm, yshift=-1cm, above=2mm, scale=0.7, black] {g};
    \path pic [xshift=2.5cm, yshift=-1cm] {bgr} node [align=center, xshift=2.45cm, yshift=-1cm, above=2mm, scale=0.7, black] {h};
    \path pic [xshift=3.5cm, yshift=-1cm] {rgb} node [align=center, xshift=3.5cm, yshift=-1cm, above=2.1mm, scale=0.7, black] {i};
    \path pic [xshift=1cm, yshift=-2cm] {grb} node [align=center, xshift=1cm, yshift=-2cm, above=2mm, scale=0.7, black] {j};
    \path pic [xshift=2cm, yshift=-2cm] {bgr} node [align=center, xshift=1.8cm, yshift=-2.1cm, above=2mm, scale=0.7, black] {k};
    \path pic [xshift=3cm, yshift=-2cm] {rbg} node [align=center, xshift=3.2cm, yshift=-2.1cm, above=2mm, scale=0.7, black] {l};
    \path pic [xshift=4cm, yshift=-2cm] {rbg} node [align=center, xshift=4.1cm, yshift=-2cm, above=2mm, scale=0.7, black] {m};
    \path pic [xshift=5cm, yshift=-2cm] {rbg} node [align=center, xshift=5cm, yshift=-2cm, above=2mm, scale=0.7, black] {n};
    \path pic [xshift=0.5cm, yshift=-3cm] {grb} node [align=center, xshift=0.45cm, yshift=-3cm, above=2mm, scale=0.7, black] {o};
    \path pic [xshift=1.5cm, yshift=-3cm] {grb} node [align=center, xshift=1.3cm, yshift=-3.1cm, above=2mm, scale=0.7, black] {p};
    \path pic [xshift=2.5cm, yshift=-3cm] {gbr} node [align=center, xshift=2.7cm, yshift=-3.1cm, above=2mm, scale=0.7, black] {q};
    \path pic [xshift=3.5cm, yshift=-3cm] {gbr} node [align=center, xshift=3.7cm, yshift=-3.07cm, above=2mm, scale=0.7, black] {r};
    \draw (0.25,0) -- (0.75,0);
    \draw (1.25,0) -- (1.75,0);
    \draw (3.25,0) -- (3.75,0);
    \draw (0,-0.25) -- (0.5,-0.75);
    \draw (1,-0.25) -- (0.5,-0.75);
    \draw (2,-0.25) -- (1.5,-0.75);
    \draw (2,-0.25) -- (2,-1.75);
    \draw (2,-0.25) -- (2.5,-2.75);
    \draw (3,-0.25) -- (2.5,-0.75);
    \draw (3,-0.25) -- (3.5,-0.75);
    \draw (4,-0.25) -- (3.5,-0.75);
    \draw (0.75,-1) -- (1.25,-1);
    \draw (1.75,-1) -- (2.25,-1);
    \draw (1.5,-1.25) -- (1,-1.75);
    \draw (1.5,-1.25) -- (1.5,-2.75);
    \draw (1.5,-1.25) -- (2,-1.75);
    \draw (1.5,-1.25) -- (3,-1.75);
    \draw (2.5,-1.25) -- (2,-1.75);
    \draw (2.5,-1.25) -- (3,-1.75);
    \draw (3.5,-1.25) -- (1.5,-2.75);
    \draw (3.5,-1.25) -- (4,-1.75);
    \draw (2.25,-2) -- (2.75,-2);
    \draw (3.25,-2) -- (3.75,-2);
    \draw (4.25,-2) -- (4.75,-2);
    \draw (1,-2.25) -- (0.5,-2.75);
    \draw (1,-2.25) -- (1.5,-2.75);
    \draw (3,-2.25) -- (1.5,-2.75);
    \draw (2,-2.25) -- (2.5,-2.75);
    \draw (3,-2.25) -- (3.5,-2.75);
    \draw (4,-2.25) -- (3.5,-2.75);
    \draw (2.75,-3) -- (3.25,-3);
\end{tikzpicture}
    \caption{A rainbow graph.}
    \label{}
\end{subfigure}
\begin{subfigure}[t]{0.45\textwidth}
    \centering
    \begin{tikzpicture}
\path pic {brg} node [align=center, above=2mm, scale=0.7, black] {a\\$i=2$};
\path pic [xshift=1cm] {brg} node [align=center, xshift=1cm, above=2mm, scale=0.7, black] {b, f\\$i=1$};
\path pic [xshift=2cm] {brg} node [align=center, xshift=2cm, above=2mm, scale=0.7, black] {c, g\\$i=0$};
\path pic [xshift=3.5cm] {bgr} node [align=center, xshift=3.5cm, above=2mm, scale=0.7, black] {h, k\\$i=0$};
\path pic [xshift=0.5cm, yshift=-1cm] {rgb} node [align=center, xshift=0.5cm, yshift=-1cm, above=2mm, scale=0.7, black] {e} node [align=center, xshift=0.5cm, yshift=-1cm, below=2mm, scale=0.7, black] {$i=1$};
\path pic [xshift=1.5cm, yshift=-1cm] {rgb} node [align=center, xshift=1.5cm, yshift=-1cm, above=2mm, scale=0.7, black] {d, i} node [align=center, xshift=1.5cm, yshift=-1cm, below=2mm, scale=0.7, black] {$i=0$};
\path pic [xshift=4cm, yshift=-1cm] {rbg} node [align=center, xshift=4cm, yshift=-1cm, above=2mm, scale=0.7, black] {l, m} node [align=center, xshift=4cm, yshift=-1cm, below=2mm, scale=0.7, black] {$i=0$};
\path pic [xshift=5cm, yshift=-1cm] {rbg} node [align=center, xshift=5cm, yshift=-1cm, above=2mm, scale=0.7, black] {n} node [align=center, xshift=5cm, yshift=-1cm, below=2mm, scale=0.7, black] {$i=1$};
\path pic [xshift=1cm, yshift=-2cm] {grb} node [align=center, xshift=1cm, yshift=-2cm, below=2mm, scale=0.7, black] {o\\$i=1$};
\path pic [xshift=2cm, yshift=-2cm] {grb} node [align=center, xshift=2cm, yshift=-2cm, below=2mm, scale=0.7, black] {j, p\\$i=0$};
\path pic [xshift=3.5cm, yshift=-2cm] {gbr} node [align=center, xshift=3.5cm, yshift=-2cm, below=2mm, scale=0.7, black] {q, r\\$i=0$};

\draw (0.25,0) -- (0.75,0);
\draw (1.25,0) -- (1.75,0);
\draw (0.75,-1) -- (1.25,-1);
\draw (4.25,-1) -- (4.75,-1);
\draw (1.25,-2) -- (1.75,-2);

\draw  (2.25,0) -- (3.25,0);
\draw  (2.25,0) -- (2.25,-2);
\draw  (2.25,0) -- (3.75,-1);
\draw  (2.25,0) -- (3.25,-2);

\draw  (1.75,-1) -- (3.75,-1);
\draw  (1.75,-1) -- (3.25,0);
\draw  (1.75,-1) -- (2.25,-2);

\draw  (2.25,-2) -- (3.25,-2);
\draw  (2.25,-2) -- (3.75,-1);

\draw  (3.25,-2) -- (3.25,0);
\draw  (3.25,-2) -- (3.75,-1);

\draw  (3.25,0) -- (3.75,-1);

\end{tikzpicture}
    \caption{The boundary rainbow graph.}
    \label{fig: rainbow b}
\end{subfigure}
\caption{A rainbow graph and its corresponding boundary graph. Each vertex represents a dataset and neighboring datasets are connected by an edge. The function output space is represented by colors \texttt{blue}, \texttt{red}, and \texttt{green}. Each dataset has a color preference, represented by the ordering inside the vertex. For example, vertex $a$ prefers \texttt{blue} to \texttt{red} and \texttt{red} to \texttt{green}. We call each such color ordering a rainbow. In this context, a DP mechanism is a probability distribution on colors of each vertex. In (b), we show the boundary rainbow graph of the rainbow graph shown in (a), as described in Definition~\ref{def:boundarymorph}. In Theorem \ref{theo: general to line} we show how, for boundary homogeneous rainbow graphs defined in Definition \ref{def:boundaryhomogeneity}, optimal $\varepsilon$-DP mechanisms on (a) can be retrieved from optimal ones on (b), which then can be obtained from Theorems~\ref{theo:resultforblue} and $\ref{theo:resultforred}$.} 
\label{fig: rainbow}
\vspace{-0.45cm}
\end{figure*}
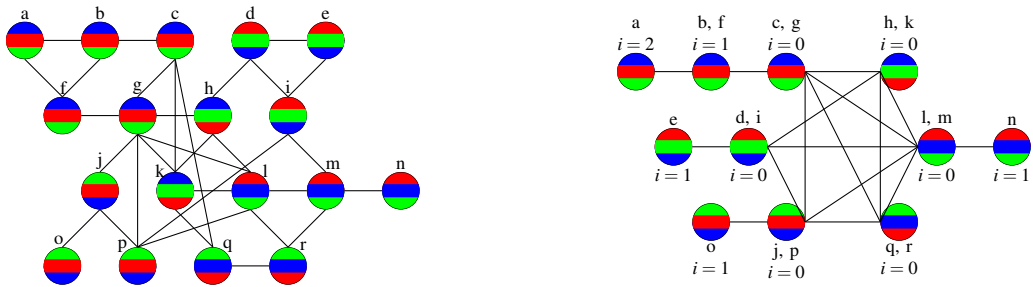

\section{Setting}
We denote by $(\mathcal{D}, \sim)$ a family of datasets together with a symmetric neighborhood relationship, where $d,d' \in \mathcal{D}$ are neighbors if $d \sim d'$. We consider a finite output space $\mathcal{V}$. Each dataset $d \in \mathcal{D}$ has an ordered preference for the elements of $\mathcal{V}$, captured by what we call a \emph{rainbow} that represents each preference order.

\begin{definition}
    Let $\mathcal{V}$ be a finite output space. A rainbow on $\mathcal{V}$ is a total ordering of $\mathcal{V}$. We denote a rainbow as a permutation vector $c \in \sym(\mathcal{V})$, where $\sym(\mathcal{V})$ is the set of all permutations of $\mathcal{V}$.
\end{definition}

Then, the preference of a dataset is captured by the \emph{preference function} $f:\mathcal{D} \rightarrow \sym(\mathcal{V})$ which assigns a rainbow to each dataset $d \in \mathcal{D}$. Thus, if $f(d) = (\texttt{blue,red, green})$, then the dataset $d \in \mathcal{D}$ prefers \texttt{blue} to \texttt{red} and \texttt{red} to \texttt{green}. The goal is to construct a random function $\mathcal{M}:\mathcal{D} ~\to ~\mathcal{V}$ that for each dataset $d \in \mathcal{D}$ randomly outputs an element of $\mathcal{V}$ such that for a given DP constraint a certain utility function is maximized. As commonly done in the DP literature, we refer to the random function as a \emph{mechanism}. A mechanism is differentially private if the distribution of its output on neighboring datasets are approximately indistinguishable, as we formalize next.

\begin{definition}[\hspace{1sp}\cite{DworkBook}]\label{def:DPdef}
Let $\varepsilon$ be a non-negative real number. Then, a mechanism $\mathcal{M}: \mathcal{D}\rightarrow \mathcal{V}$ is $\varepsilon$-DP if for any $d \sim d'$ and $\mathcal{S} \subseteq \mathcal{V}$, it holds that $\Pr[\mathcal{M}(d) \in \mathcal{S}] \leq e^\varepsilon \Pr[\mathcal{M}(d') \in \mathcal{S}]$. We denote the set of all $\varepsilon$-DP mechanisms by $\mathfrak{M}$.
\end{definition}

For finite output space $\mathcal{V}$ it suffices to consider subsets $\mathcal{S} \subseteq \mathcal{V}$ with $|\mathcal{S}|=1$, i.e., Definition~\ref{def:DPdef} holds, if and only if, $\Pr[\mathcal{M}(d)=v]\leq e^{\varepsilon}\Pr[\mathcal{M}(d')=v]$ for every $d \sim d'$ and $v \in \mathcal{V}$. 

Performance of a mechanism is measured through a utility function $U:\mathfrak{M} \rightarrow \mathbb{R}$, where $U[\mathcal M] \geq U[\mathcal M']$ means that the mechanism $\mathcal M$ outperforms $\mathcal M'$. In this work, we consider utility functions that agree with the preference function $f:\mathcal{D} \rightarrow \sym(\mathcal{V})$, i.e., all things equal, it is preferable for a dataset $d \in \mathcal{D}$ to output a color it prefers according to its rainbow $f(d) \in \sym(\mathcal{V})$.

\begin{definition} \label{def:order_reasonable}
Let $\preceq$ be the lexicographical ordering on the probability simplex $\Delta(\mathcal{V}) \!=\! \{x \!\in\! [0,1]^{|\mathcal{V}|}\! : x_1\!+\!\cdots\!+\!x_{|\mathcal{V}|} =~1\}$. For every mechanism $\mathcal{M} \in \mathfrak{M}$ and dataset $d \in \mathcal{D}$, let $\vec{\mathcal{M}}(d) \!\!~\in~\!\! \Delta(\mathcal{V})$ be the vector with coordinates $\vec{\mathcal{M}}_k = \Pr[\mathcal{M}(d) = f(d)_k]$. Then, a mechanism $\mathcal{M}\in \mathfrak{M}$ dominates another mechanism $\mathcal{M}'\in \mathfrak{M}$ if for every dataset $d \in \mathcal{D}$, $\vec{\mathcal{M}}(d) \succeq \vec{\mathcal{M}'}(d)$. Moreover, we define a utility function $U:\mathfrak{M} \rightarrow \mathbb{R}$ as \textit{order reasonable} if whenever a mechanism $\mathcal{M}\in \mathfrak{M}$ dominates another mechanism $\mathcal{M}'\in \mathfrak{M}$, then $U[\mathcal{M}] \geq U[\mathcal{M}']$.
\end{definition}

The notion of domination in Definition~\ref{def:order_reasonable} induces a partial order on the set $\mathfrak{M}$ of all $\varepsilon$-DP mechanisms. When a mechanism $\mathcal M$ dominates $\mathcal M'$, it means that $\mathcal M$ outperforms $\mathcal M'$ for any order reasonable utility. In this setting, we say a mechanism is \emph{optimal} if no other mechanism dominates it.

As in \cite{RafnoDP2colorPaper}, we represent a family of datasets together with their neighboring relation $(\mathcal{D},\sim, f)$ by a simple graph, where the vertices are the datasets in $\mathcal{D}$ and there is an edge between $d,d' \in \mathcal{D}$ if and only if they are neighbors, i.e., $d \sim d'$.

\begin{definition}[\hspace{1sp}\cite{RafnoDP2colorPaper}]\label{def:graphmorphism}
A morphism between $(\mathcal{D}_1, \overset{1}{\sim})$ and $(\mathcal{D}_2, \overset{2}{\sim})$ is a function $g: (\mathcal{D}_1, \overset{1}{\sim}) \rightarrow (\mathcal{D}_2, \overset{2}{\sim})$ such that ${d\overset{1}{\sim}d'}$ implies in either $g(d)\overset{2}{\sim}g(d')$ or $g(d)=g(d')$ for every $d,d' \in \mathcal{D}_1$.
\end{definition}

An example of a morphism is shown in Fig. \ref{fig: rainbow}. A morphism $g: (\mathcal{D}_1, \overset{1}{\sim}) \rightarrow (\mathcal{D}_2, \overset{2}{\sim})$ allows to transport $\varepsilon$-DP mechanisms from its codomain to its domain. 

\begin{theorem}[\hspace{1sp}\cite{RafnoDP2colorPaper}]\label{theo:pullbackworks}
Let $g: (\mathcal{D}_1, \overset{1}{\sim}) \rightarrow (\mathcal{D}_2, \overset{2}{\sim})$ be a morphism and $\mathcal{M}_2: \mathcal{D}_2 \rightarrow \mathcal{V}$ be an $\varepsilon$-DP mechanism on $(\mathcal{D}_2, \overset{2}{\sim})$. Then, the mechanism $\mathcal{M}_1: \mathcal{D}_1 \rightarrow \mathcal{V}$ given by the pullback operation $\mathcal{M}_1 = \mathcal{M}_2 \circ g$ is an $\varepsilon$-DP mechanism on $\mathcal{D}_1$.
\end{theorem}

\section{Optimal Rainbow Differential Privacy}

In \cite{RafnoDP2colorPaper}, DP schemes were interpreted as randomized graph colorings. In that setting, each dataset's preference was characterized by a single color. In general, for larger output spaces, each dataset has a corresponding rainbow according to its ordering preference. Thus, we call the triple $(\mathcal{D},\sim,f)$ a \emph{rainbow graph}, where $\mathcal{D}$ is the family of datasets, $\sim$ is the neighborhood relationship, and $f: \mathcal{D} \rightarrow \sym (\mathcal{V})$ is the preference function. We define a morphism $g\!:\!(\mathcal{D}_1, \overset{1}{\sim},f_1) \!\rightarrow\! (\mathcal{D}_2,\overset{2}{\sim},f_2)$ as \emph{rainbow-preserving} if $f_1 = f_2 \circ g$. Indeed, the morphism in Fig. \ref{fig: rainbow} is rainbow-preserving. We consider the following topological notions.

\begin{definition}
    Let $(\mathcal{D},\sim,f)$ be a rainbow graph. Then, for every $c \in \sym(\mathcal{V})$, we denote $C = \{ d\in \mathcal{D} : f(d)=c \}$. The interior of $C$ is the set $\interior{C} = \{ d\in C : d \sim d' \Rightarrow d' \in C \}$ and its boundary is the set $\partial C = C - \interior{C}$.
\end{definition}

Next, we define a homogeneity condition for DP mechanisms, as defined in \cite{RafnoDP2colorPaper} for binary functions.

\begin{definition}\label{def:boundaryhomogeneity}
A mechanism $\mathcal{M}: \mathcal{D} \rightarrow \mathcal{V}$ is \emph{boundary homogeneous} if, for every rainbow $c \in \sym(\mathcal{V})$, it holds that any two boundary datasets $d,d' \in \partial C$ satisfy $\Pr[\mathcal{M}(d) = v] = \Pr[\mathcal{M}(d') = v]$ for every $v \in \mathcal{V}$.
\end{definition}

Our next result shows that optimal boundary homogeneous DP mechanisms are fully characterized by their values at the boundary set.

\begin{theorem} \label{teo: unique optimal}
Let $(\mathcal{D},\sim,f)$ be a rainbow graph and, for every rainbow $c \in \sym(\mathcal{V})$ and $d \in \partial C$, let $(m_1,m_2,\ldots, m_{|\mathcal{V}|})=\vec{m}\in \Delta (\mathcal{V})$ be a fixed and homogeneous probability distribution. Then, there exists at most one optimal boundary homogeneous $\varepsilon$-DP mechanism $\mathcal{M}:\mathcal{D} \rightarrow \mathcal{V}$ such that $\Pr[\mathcal{M}(d) ~=~ v_k] = m_k$, for every $d \in \partial C$.
\end{theorem}

\begin{proof}
We assume that there exists at least one boundary homogeneous mechanism satisfying the $\varepsilon$-DP constraints, otherwise the result trivially holds. Let $c\in \sym(\mathcal{V})$ be a rainbow and consider the set $C$. We denote $C_d^k = \Pr[\mathcal{M}(d)~=~c_k]$. Then, the $\varepsilon$-DP constraints in $C$ are equivalent to the following statement: for every $k \in [1: |\mathcal{V}|]$ and $d \sim d'$ with $d,d \in C$, it holds that $C_d^k \geq 0$, $\sum_{k=1}^{|\mathcal{V}|} C_d^k =1$, $C_d^k \leq e^\varepsilon \cdot C_{d'}^k$, and $C_{d'}^k \leq e^\varepsilon \cdot C_{d}^k$.

Consider the highest priority color $c_1$ in $C$. If $d\sim d'$ are two neighboring datasets in $C$, the dataset $d'$ imposes two upper bounds on $d$, namely $C_{d}^1 \leq e^\varepsilon \cdot C_{d'}^1$ and 
\begin{align}
    C_{d}^1 = 1 - \sum_{k=2}^{|\mathcal{V}|} C_{d}^k \leq 1 - e^{- \varepsilon} \sum_{k=2}^{|\mathcal{V}|} C_{d'}^k = \frac{e^\varepsilon - 1 + C_{d'}^1}{e^\varepsilon}.
\end{align}
Since both bounds are non-decreasing in $C_{d'}^1$, it holds that all $C_{d}^1$ with $d \in \interior{C}$ can be simultaneously maximized. Denote these maximums by $\Bar{C}_{d}^1$. 

Now, when we consider the second highest priority color $c_2$ in $C$, if we set $C_{d}^1 = \Bar{C}_{d}^1$, then the $\varepsilon$-DP constraints in $C$ are equivalent to the following statement: for every $k \in [1:~|\mathcal{V}|]-~\{1\}$ and $d \sim d'$ with $d,d \in C$, it holds that $C_d^k \geq 0$, $\sum_{k=2}^{|\mathcal{V}|} C_d^k =1-\Bar{C}_{d}^1$, $C_d^k \leq e^\varepsilon \cdot C_{d'}^k$, and $C_{d'}^k \leq e^\varepsilon \cdot C_{d}^k$. Analogously to the case of $c_1$, if $d\sim d'$ are two neighboring datasets in $C$, the dataset $d'$ imposes two upper bounds on $d$ with respect to $C_{d'}^2$, both of which are non-decreasing in it. Thus, all the $C_{d}^2$ with $d \in \interior{C}$ can be simultaneously maximized. Denote these maximums by $\Bar{C}_{d}^2$. 

Repeating this argument for every color in $\mathcal{V}$ and then for every rainbow $c$ in the rainbow graph, we obtain a unique optimal boundary homogeneous $\varepsilon$-DP mechanism. 
\end{proof}

Another key notion we use is that of the line graph.

\begin{definition}
    Let $c \in \sym(\mathcal{V})$ be a rainbow and $n \in \mathbb{N}$. The $(n,c)$-line is the rainbow graph $(\mathcal{D},\sim,f)$ with datasets $\mathcal{D} = [1:~(n-1)]$, neighboring relation $i \sim j$ if $|i-j|=1$, and preference function $f(d)=c$ for every $d \in \mathcal{D}$.
\end{definition}

The last notion we need for Theorem \ref{theo: general to line} is that of the boundary rainbow graph of a rainbow graph.

\begin{definition}\label{def:boundarymorph}
The boundary morphism of a rainbow graph $(\mathcal{D},\sim,f)$ is the morphism $g_\partial : \mathcal{D} \rightarrow \sym(\mathcal{V}) \times \mathbb{N}$ such that $g_\partial = [f(d),\dist (d,\partial f(d))]$. The boundary rainbow graph of $(\mathcal{D},\sim,f)$ is then the rainbow graph $(\mathcal{D}_\partial, \overset{\partial}{\sim},f_\partial)$ with datasets $\mathcal{D}_\partial = g_\partial(\mathcal{D})$, preference functions $f_\partial = f \circ g_\partial^{-1}$, and neighboring relationship where two distinct datasets are $d_1' \overset{\partial}{\sim} d_2'$ if $g_\partial^{-1} (d_1') \sim g_\partial^{-1} (d_2')$. For every rainbow $c\in \sym(\mathcal{V})$, we define $C_\partial = \{ d\in \mathcal{D}_\partial : f(d)=c \}$.
\end{definition}

Thus, the boundary rainbow graph consists of a series of line graphs, each for a different rainbow occurring in the original graph. We now show that optimal mechanisms for boundary homogeneous rainbow graphs can be obtained by pulling them back from their boundary rainbow graphs.

\begin{theorem} \label{theo: general to line}
Let $(\mathcal{D}, \sim, f)$ be a rainbow graph and  $\mathcal{M}_\partial: \mathcal{D}_\partial \rightarrow \mathcal{V}$ be the optimal $\varepsilon$-DP mechanism on its boundary graph subject to some fixed boundary probabilities. Then, the pullback $\mathcal{M} = \mathcal{M}_\partial \circ g_\partial$ is the optimal boundary homogeneous $\varepsilon$-DP mechanism subject to the same boundary probabilities.
\end{theorem}

\begin{proof}
From Theorem \ref{theo:pullbackworks}, it follows that the morphism $g_\partial: \mathcal{D} \rightarrow \mathcal{D}_{\partial}$ induces an $\varepsilon$-DP mechanism on $\mathcal{D}$ defined by $\mathcal{M}_{\partial}\circ~g_\partial$. This mechanism is clearly boundary homogeneous. From Theorem \ref{teo: unique optimal}, it follows that there is a unique optimal boundary homogeneous $\varepsilon$-DP mechanism on $\mathcal{D}$. Denote this mechanism by $\mathcal{M}$. We next show that $\mathcal{M} = \mathcal{M}_{\partial} \circ g_\partial$.

Let $C \subseteq \mathcal{D}$ be a subset of datasets with the same preference function. It follows from the optimality of $\mathcal{M}$ that $\vec{\mathcal{M}}(d) \succeq \overrightarrow{\mathcal{M}_\partial \circ g_\partial }(d)$. Let $d_0$ be the closest dataset to $d$ belonging to $\partial C$. Let $G = \{d, d_{\dist(d,\partial B)-1}, \ldots, d_0 \}$ be a set of datasets which forms a shortest path from $d$ to $d_0$. Since $g_{\partial} \restrictedto{G}$ is injective, it has a left inverse, which we denote as $h: g_\partial (G) ~\rightarrow ~G$. But $h$ is a morphism and, therefore, from Theorem \ref{def:graphmorphism} $\mathcal{M}\circ h$ is an $\varepsilon$-DP mechanism on $\mathcal{D}_\partial$. Then, since $\mathcal{M}_\partial$ is the optimal mechanism on $\mathcal{D}_\partial$, it follows that $\mathcal{M}_\partial \restrictedto{g_\partial (G)}$ is the optimal mechanism on $g_\partial (G)$. Thus, $\overrightarrow{\mathcal{M}_\partial \circ g_\partial }(d) \succeq \vec{\mathcal{M}}(d) $, which implies $\vec{\mathcal{M}}(d) = \overrightarrow{\mathcal{M}_\partial \circ g_\partial }(d)$.
\end{proof}

Since the boundary rainbow graph consists of a series of line graphs, the problem of finding optimal mechanisms can be reduced to finding them for line graphs.
\section{Optimal Line Graphs for $3$-Colored Rainbows}
In this section we present closed form expressions for the optimal $\varepsilon$-DP mechanisms over $(n,c)$-line graphs for $3$-colored rainbows. We denote the output space by $\mathcal{V}=\{\texttt{blue, red, green}\}$ and consider, without loss of generality, the rainbow $c = (\texttt{blue, red, green})$. To simplify notation, we denote by $B_i = \Pr[\mathcal{M}(i)=~\texttt{blue}]$, $R_i = \Pr[\mathcal{M}(i)~=~\texttt{red}]$, and $G_i = \Pr[\mathcal{M}(i)=\texttt{green}]$.

For every $a\in\mathbb{R}$, we define $\{a\}^{+}=\max\{a,0\}$. We also define the following index thresholds.
\begin{align}
    &\tau_B =\Bigg\lfloor\bigg\{-\frac{\ln\big(B_0\cdot(e^\varepsilon +1)\big)}{\varepsilon} + 1 \bigg\}^{+}\Bigg\rfloor \label{eq:deftauB}\\
   &\tau_R = \Bigg\lfloor\bigg\{-\frac{\ln\big((B_0+R_0)\cdot(e^\varepsilon +1)\big)}{\varepsilon} + 1\bigg\}^{+}\Bigg\rfloor\label{eq:deftauR}.
\end{align}

Consider $d_i,d_{i+1}\in B$ such that $d_i\sim d_{i+1}$ for all $i=[0:(n-1)]$, then all active $\varepsilon$-DP constraints for the ternary output space $\mathcal{V}$ are
\noindent\begin{minipage}{.5\linewidth}
\begin{align}
&B_i,\; R_i,\; G_i \geq 0\label{eq:DPcondallnonnegative}\\
&B_i \leq e^{\varepsilon}\cdot B_{i+1}\label{eq:DPcondupperforBi}\\
&R_i \leq e^{\varepsilon}\cdot R_{i+1}\label{eq:DPcondupperforRi}\\
&G_i \leq e^{\varepsilon}\cdot G_{i+1}\label{eq:DPcondupperforGi}
\end{align}
\end{minipage}%
\begin{minipage}{.5\linewidth}
\begin{align}
&B_i + R_i + G_i = 1\label{eq:DPcondsumprob}\\
&B_{i+1} \leq e^{\varepsilon}\cdot B_i\label{eq:DPcondupperforBiplus1}\\
&R_{i+1} \leq e^{\varepsilon}\cdot R_i\label{eq:DPcondupperforRiplus1}\\
&G_{i+1} \leq e^{\varepsilon}\cdot G_i.\label{eq:DPcondupperforGiplus1}
\end{align}
\end{minipage}
\\

Now, we denote the probabilities for outputs of an optimal $\epsilon$-DP mechanism on dataset $i$ by $B^*_{i}$, $R^*_{i}$, and $G^*_{i}$, and consider the problem of finding a closed form for them as a function of the values at $B^*_{0}$, $R^*_{0}$, and $G^*_{0}$. We begin with $B^*_{i}$.

\begin{theorem} \label{theo:resultforblue}
Let $B_0,R_0,G_0 \in \Delta(\mathcal{V})$. Then, the unique optimal boundary homogeneous $\varepsilon$-DP mechanism for the $(n,c)$-line with rainbow $c = (\texttt{blue, red, green})$ such that $B^*_{0}=B_{0}$, $R^*_{0}=R_{0}$, and $G^*_{0}=G_{0}$ satisfies 
\begin{align}
B^*_i =\begin{cases}
&e^{i\varepsilon}B_0 \qquad\qquad\qquad\quad\;\;\;\;\; \text{if} \quad 1\leq i \leq \tau_B \\
&1 \!-\! e^{(\tau_B-i)\varepsilon} \!+\! e^{(2\tau_B-i)\varepsilon} B_0 \;\;\; \text{if} \quad \tau_B< i< n \label{eq:Bistarresult}.
\end{cases}
\end{align}
\end{theorem}

\begin{proof}
Since \texttt{blue} is the first element in the rainbow, it should be maximized first. There are two achievable upper bounds on $B_{i+1}$ imposed to satisfy the $\varepsilon$-DP constraints in (\ref{eq:DPcondallnonnegative})-(\ref{eq:DPcondupperforGiplus1}). The first one is characterized by (\ref{eq:DPcondupperforBiplus1}), and the second is characterized jointly by (\ref{eq:DPcondupperforRi})-(\ref{eq:DPcondsumprob}) such that 
\begin{align}
    B_{i+1} \leq 1 - e^{-\varepsilon} (R_i +G_i) = 1-e^{-\varepsilon}(1-B_i).\label{eq:secondoutputforBiplus1}
\end{align}
The two upper bounds in (\ref{eq:DPcondupperforBiplus1}) and (\ref{eq:secondoutputforBiplus1}) on $B_{i+1}$ are equal when $B_i$ is equal to the threshold
\begin{align}
\overline{T}\triangleq {(e^{\varepsilon}+1)}^{-1}\label{eq:Bithreshold}.
\end{align}
Furthermore, by (\ref{eq:DPcondupperforBiplus1}) we have
\begin{align}
    B^*_{i}= e^{\varepsilon} B_{i-1} \qquad\text{if}\qquad B_{i-1}\leq \overline{T}\label{eq:Biplus1maximumfirststage}
\end{align}
which can be represented equivalently as
\begin{align}
    B^*_{i}= e^{i\varepsilon} B_0\qquad \text{if}\qquad  B_{i-1}^*\leq \overline{T}\label{eq:Biplus1maximumfirststageequiv}
\end{align}
since $B^*_{i-1}$ is non-decreasing in $i$ if $B_{i-1}^*\leq \overline{T}$ and where 
\begin{align}
    B_{i-1}^* = e^{{(i-1)}\varepsilon}\cdot B_0 \leq \overline{T} \Longleftrightarrow i \leq -\frac{\ln\big(B_0(e^\varepsilon +1)\big)}{\varepsilon} + 1.\label{eq:BbartotauB}
\end{align}
Thus, from the definition of $\tau_B$ in (\ref{eq:deftauB}), the first case in (\ref{eq:Bistarresult}) holds. Next, by (\ref{eq:secondoutputforBiplus1}) we obtain
\begin{align}
    B^*_{i}  = 1-e^{-\varepsilon}(1-B_{i-1})\qquad\text{if}\qquad B_{i-1}> \overline{T}\label{eq:Biplus1maximumsecondstage}
\end{align}
which is equivalent to
\begin{align}
    B^*_{i}= 1-e^{-(i-\tau_B)\varepsilon}(1-B^*_{\tau_B})\qquad \text{if}\qquad  \tau_B<i\leq n-1\label{eq:Biplus1maximumsecondstageequiv}.
\end{align}
Inserting $B^*_{\tau_B}=e^{\tau_B\varepsilon}\cdot B_0$ into (\ref{eq:Biplus1maximumsecondstageequiv}), we have
\begin{align}
    B^*_{i}= 1\!-\!e^{-(i-\tau_B)\varepsilon}\!+\!e^{-(i-2\tau_B)\varepsilon}B_0\qquad \text{if}\quad  \tau_B\!< i\!\leq\! n\!-\!1
\end{align}
which proves the second case in (\ref{eq:Bistarresult}).
\end{proof}

We now give an expression for $R^*_{i}$.

\begin{theorem} \label{theo:resultforred}
Let $B_0,R_0,G_0 \in \Delta(\mathcal{V})$. Then, the unique optimal boundary homogeneous $\varepsilon$-DP mechanism for the $(n,c)$-line with rainbow $c = (\texttt{blue, red, green})$ such that $B^*_{0}=B_{0}$, $R^*_{0}=R_{0}$, and $G^*_{0}=G_{0}$ satisfies 

\begin{equation}
R_i^* = 
    \begin{cases}
    e^{i\varepsilon} R_0  \qquad& \text{if} \quad 1\leq i \leq \tau_R\\
    R_{\text{Mid}}(i)  \qquad& \text{if} \quad \tau_R < i \leq \tau_B \\
    e^{-(i-\tau_B)\varepsilon} R_{\text{Mid}}(\tau_B) \qquad& \text{if} \quad \tau_B< i < n \label{eq:Ristarresult}
    \end{cases}
\end{equation}
where 
\begin{align}
&R_{\text{Mid}}(i) =1- e^{-(i-\tau_R)\varepsilon}- e^{i\varepsilon}B_0+e^{-(i-2\tau_R)\varepsilon}(B_0+R_0)\label{eq:RMidinewdef}.
\end{align}
\end{theorem}

\begin{proof}
Since \texttt{red} is the second color that should be maximized after \texttt{blue} is maximized, we analyze $R_i^{*}$ separately for the two index regimes given in (\ref{eq:Bistarresult}) that are characterized by $\tau_B\geq 0$.

First, for $1\leq i \leq \tau_B$, there are two achievable upper bounds on $R_{i+1}$ that satisfy the $\varepsilon$-DP constraints. We first have the upper bound in (\ref{eq:DPcondupperforRiplus1}), and a second upper bound follows from (\ref{eq:DPcondupperforGi}), (\ref{eq:DPcondsumprob}), (\ref{eq:Biplus1maximumfirststage}), and (\ref{eq:BbartotauB}), namely $B_{i+1}^*=e^{\varepsilon}B_i^*$ for $1\leq i \leq \tau_B-1$, such that we obtain
\begin{align}
    R_{i+1}\leq 1 -e^{\varepsilon}B_i^* - e^{-\varepsilon}(1-B_i^*-R_i).\label{eq:secondoutputforRiplus1}
\end{align}
The two upper bounds in (\ref{eq:DPcondupperforRiplus1}) and (\ref{eq:secondoutputforRiplus1}) are equal if the sum $(B_i^*+R_i)$ is equal to the threshold given in (\ref{eq:Bithreshold}). Furthermore, it follows from (\ref{eq:DPcondupperforRiplus1}) that $R^*_{i}= e^{i\varepsilon} R_{0}$ if $(B^{*}_{i-1}\!\!~+~\!\!R^*_{i-1})\leq \overline{T}$, since $(B^*_{i-1}\!\!~+~\!\!R^*_{i-1})$ is non-decreasing in $i$ if $(B^{*}_{i-1}\!\!~+~\!\!R^*_{i-1})\leq \overline{T}$. Similar to (\ref{eq:BbartotauB}), using (\ref{eq:Biplus1maximumfirststage}) and (\ref{eq:Biplus1maximumfirststageequiv}) we have $(B^{*}_{i-1}\!\!~+~\!\!R^*_{i-1})~\leq~ \overline{T}$ if and only if
\begin{align}
    i \leq -\frac{\ln\big((B_0+R_0)(e^\varepsilon +1)\big)}{\varepsilon} + 1.\label{eq:TbartotauR}
\end{align}
By defining the new index threshold $\tau_R$ as in (\ref{eq:deftauR}), the first case in (\ref{eq:Ristarresult}), i.e., $1\leq i\leq \tau_R$ case, is proved. Next, by (\ref{eq:Biplus1maximumfirststageequiv}) and (\ref{eq:secondoutputforRiplus1}) we obtain for $\tau_R< i \leq\tau_B$ that
\begin{align}
    R_{i}^* &= 1 - e^{i\varepsilon}B_0 - e^{-(i-\tau_R)\varepsilon}(1 - (B^*_{\tau_R} +R^*_{\tau_R}))\nonumber\\
&= 1 - e^{i\varepsilon}B_0 - e^{-(i-\tau_R)\varepsilon}(1 - e^{\tau_R\varepsilon}(B_0 +R_0))\nonumber\\
&=1- e^{-(i-\tau_R)\varepsilon}- e^{i\varepsilon}B_0+e^{-(i-2\tau_R)\varepsilon}(B_0+R_0)\nonumber\\
&= R_{\text{Mid}}(i)
\end{align}
where $R_{\text{Mid}}(i)$ is as defined in (\ref{eq:RMidinewdef}).

Second, for $\tau_B< i \leq n-1$, there is no achievable upper bound on $R_{i+1}$ that satisfies the $\varepsilon$-DP constraints since by (\ref{eq:secondoutputforBiplus1}) and (\ref{eq:Biplus1maximumsecondstage}) we have
\begin{align*}
    B_i^* = 1-e^{-\varepsilon}(R_{i-1}^*+G_{i-1})\qquad \text{if}\qquad  \tau_B<i\leq n-1
\end{align*}
and by (\ref{eq:DPcondsumprob})
\begin{align}
    B_i^* = 1- (R_i^*+G_i)\qquad\text{for all}\quad 1\leq i\leq n-1.
\end{align}
Thus, to maximize the probability $B_i^*$ of the first color \texttt{blue} for $\tau_B< i \leq n-1$ in which it cannot grow exponentially, one should minimize both $R_i$ and $G_i$ such that $\varepsilon$-DP constraints are satisfied since this is equivalent to minimizing $(R_i+G_i)$. By (\ref{eq:DPcondupperforRi}), we obtain
\begin{align}
    R_i^* = e^{-\varepsilon}R_{i-1}^*\qquad\text{if}\qquad \tau_B< i\leq n-1
\end{align}
which can be expressed equivalently as
\begin{align}
    R_i^* \!=\! e^{-(i-\tau_B)\varepsilon}R_{\tau_B}^*\!=\!e^{-(i-\tau_B)\varepsilon}R_{\text{Mid}}(\tau_B)
\end{align}
if $\tau_B < i\leq (n-1)$, which proves the last case in (\ref{eq:Ristarresult}).
\end{proof}

The expression for $G^*_{i}$ follows directly from the previous theorems together with the probability constraint.

\begin{corollary}\label{cor:resultforgreen}
Let $B_0,R_0,G_0 \in \Delta(\mathcal{V})$. Then, the unique optimal boundary homogeneous $\varepsilon$-DP mechanism for the $(n,c)$-line with rainbow $c = (\texttt{blue, red, green})$ such that $B^*_{0}~=~B_{0}$, $R^*_{0}=R_{0}$, and $G^*_{0}=G_{0}$ satisfies 
\begin{align}
    G_i^* = \begin{cases}
     1 - e^{i\varepsilon}(B_0+R_0) & \text{if }\; 1\leq i \leq \tau_R\\
    e^{-(i-\tau_R)\varepsilon}\cdot(1\!-\!e^{\tau_R\varepsilon}(B_0+R_0)) & \text{if }\; \tau_R < i< n.\label{eq:Gistarresult}
  \end{cases}
\end{align}
\end{corollary}

As another corollary, we recover the expressions for two colors given in \cite{RafnoDP2colorPaper}.

\begin{corollary}
Theorems~\ref{theo:resultforblue} and \ref{theo:resultforred} generalize \cite[Theorem~4]{RafnoDP2colorPaper} for $\delta=0$.
\end{corollary}

\begin{proof}
For two colors we have $B_0 + R_0 = 1$, so we obtain $\tau_R ~=~ 0$ since the output of the $\{\cdot\}^{+}$ operation in (\ref{eq:deftauR}) is $0$ for this case. Furthermore, for the 2-color problem in (\ref{eq:Ristarresult}) we obtain $R_{\text{Mid}}(i) = 1 - e^{i\varepsilon}B_0$ and $R_{\text{Ter}}(i) = e^{-(i-\tau_B)\varepsilon}R_{\text{Mid}}(\tau_B)$, which recovers the 2-color problem results in \cite[Theorem~4]{RafnoDP2colorPaper} when $\delta=0$ is imposed to provide $\varepsilon$-DP. 
\end{proof}

\begin{figure*}[!t]
\centering
\begin{subfigure}[t]{0.43\textwidth}
    \centering
     \includegraphics[width=0.95\textwidth, height=0.4\textheight,keepaspectratio]{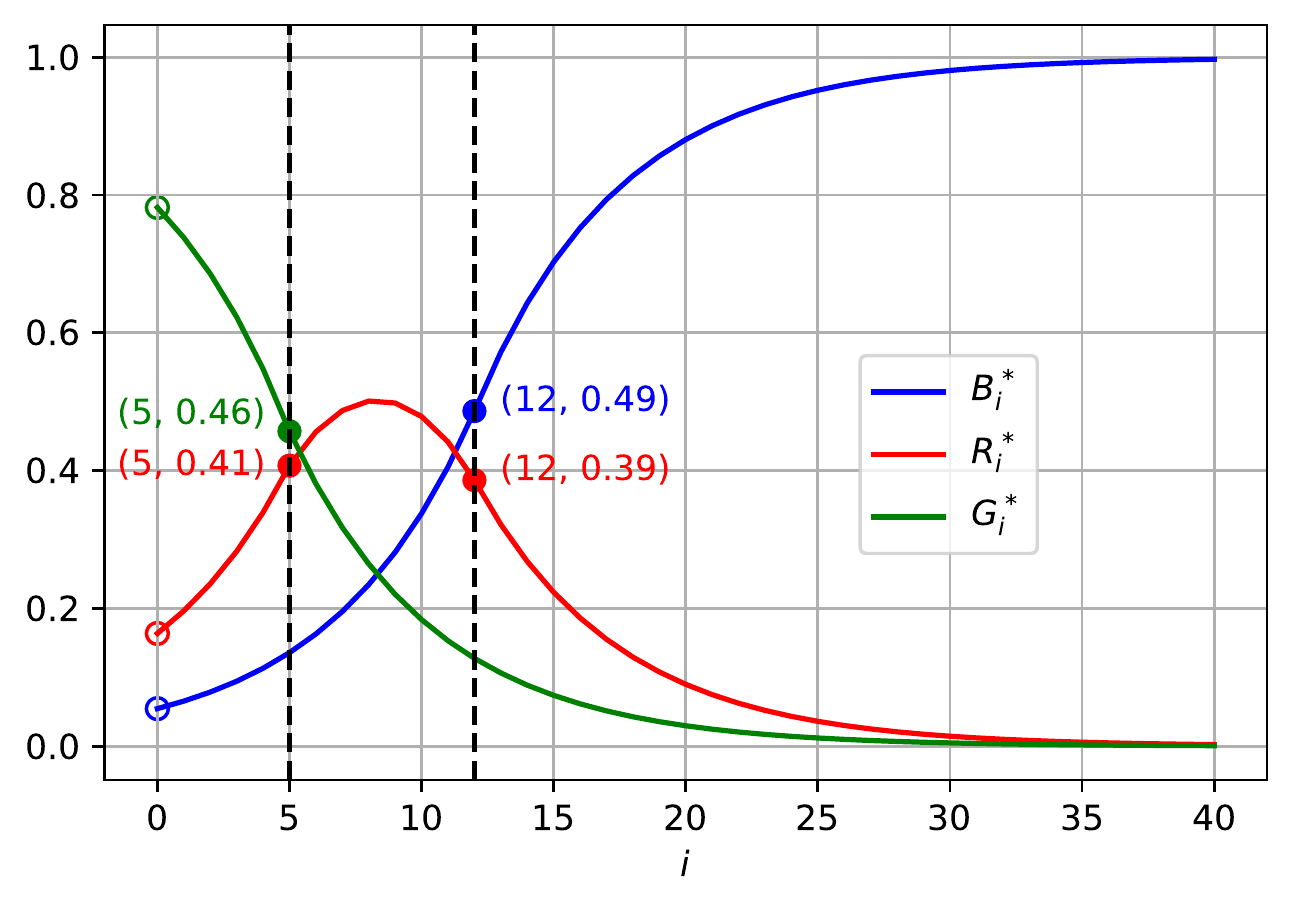}
     \vspace{-0.2cm}
     \caption{%
     Optimal $\varepsilon$-DP probabilities of being \texttt{blue, red}, and \texttt{green} vs. the distance $i$ to the \texttt{blue} boundary, given $(B_0,R_0,G_0)$=$(0.0545,0.1636,0.7818)$ and $\varepsilon=0.1823$.}
     \label{fig:Example1plot}
\end{subfigure} \quad
\begin{subfigure}[t]{0.43\textwidth}
    \centering
     \includegraphics[width=0.95\textwidth, height=0.4\textheight,keepaspectratio]{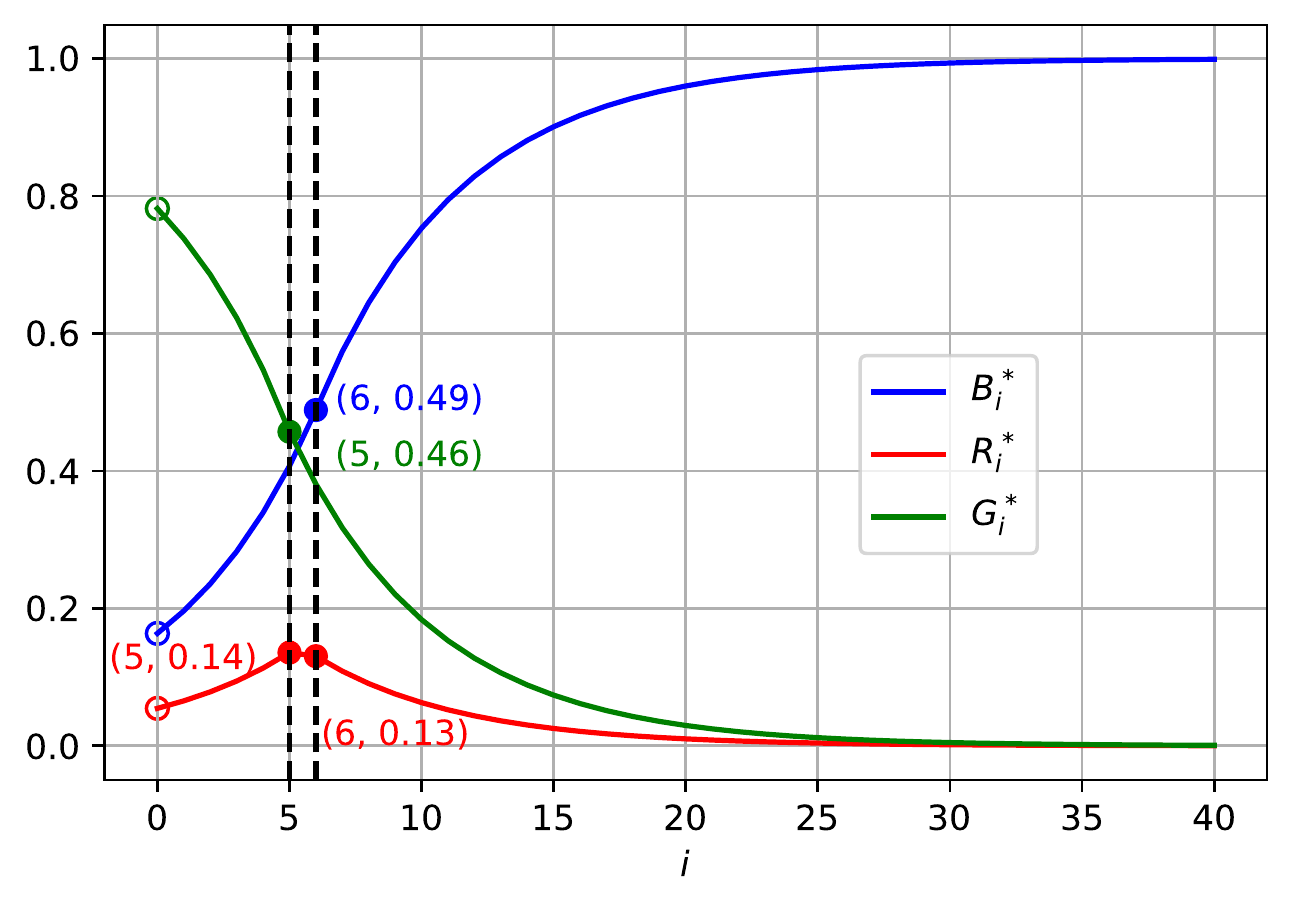}
     \vspace{-0.2cm}
     \caption{%
     Optimal $\varepsilon$-DP probabilities of being \texttt{blue, red}, and \texttt{green} vs. the distance $i$ to the \texttt{blue} boundary, given $(B_0,R_0,G_0)$=$(0.1636,0.0545,0.7818)$ and $\varepsilon=0.1823$.}
     \label{fig:Example2plot}
\end{subfigure}
\caption{Two examples that illustrate behaviors of color probabilities. The first preferred value \texttt{blue} has up to two operating regimes that are characterized by its boundary probability and $\varepsilon$. The second preferred value \texttt{red} can exhibit up to three operating regimes that are characterized additionally by the sum of boundary probabilities of \texttt{blue} and \texttt{red}.}
\label{fig: examples}
\vspace{-0.45cm}
\end{figure*} 
\section{Three Color Rainbow DP Examples}
Given $(B_0,R_0,G_0,\varepsilon)$, the unique optimal boundary homogeneous $\varepsilon$-DP mechanism for the rainbow $c = (\texttt{blue, red, green})$ is characterized by (\ref{eq:Bistarresult}), (\ref{eq:Ristarresult}), and (\ref{eq:Gistarresult}), as proved above. We now illustrate the behavior of the $(B^*_i,R^*_i,G^*_i)$ for $1\leq i\leq (n-1)$ for different boundary probabilities and DP parameters.

\begin{example}\label{ex:Example1}
 Suppose $(B_0,R_0,G_0)=(0.0545,0.1636,0.7818)$, where the order $B_0<R_0<G_0$ is the inverse of the color priority order for the rainbow $c = (\texttt{blue, red, green})$, and $\varepsilon = 0.1823$ such that $\tau_R = 5$ and $\tau_B= 12$. We plot the results of the corresponding unique optimal $\varepsilon$-DP mechanism in the rainbow $c = (\texttt{blue, red, green})$ for $i=0,1,\ldots,(n-1)$ in Fig.~\ref{fig:Example1plot}, where the optimal probabilities $(B_i^{*},R_i^{*},G_i^{*})$ for each vertex, and index thresholds $\tau_R$ and $\tau_B$ are depicted. We remark that one can choose $n$ as large as possible since the probabilities and the thresholds do not depend on the number $n$ of vertices in the rainbow $c = (\texttt{blue, red, green})$.

 For $1\leq\ i\leq \tau_R$, both \texttt{blue} and \texttt{red} increase exponentially. Then, for $\tau_R < i\leq \tau_B$, \texttt{blue} can continue increasing exponentially, whereas \texttt{red} first increases but then decreases. Finally, for $\tau_B < i\leq (n-1)$, \texttt{red} and \texttt{green} decrease exponentially, while \texttt{blue} increases slower than exponentially. 
 
 In the index range $0\leq i\leq 5$, the initial order $B_0<R_0~<~G_0$ is preserved, and in $ 5< i\leq 8$ the order is $B_i\leq G_i \leq R_i$, then in $8 < i\leq 11$ we have $G_i\leq B_i \leq R_i$, and finally in $11\!\!~ <~\!\! i\!\!~\leq~\!\! (n-1)$ the priority order $G_i\leq R_i\leq B_i$ in the rainbow $c = (\texttt{blue, red, green})$ is achieved and preserved.
\end{example}

\begin{example}\label{ex:Example2}
 Suppose next $(B_0,R_0,G_0)=(0.1636,0.0545,0.7818)$, where the order is $R_0<B_0<G_0$, and $\varepsilon = 0.1823$ such that $\tau_R =5$ and $\tau_B=6$. We plot the results of the corresponding unique optimal $\varepsilon$-DP mechanism in the rainbow $c = (\texttt{blue, red, green})$ in Fig.~\ref{fig:Example2plot}. Similar curve patterns to patterns in Fig.~\ref{fig:Example1plot} are observed in Fig.~\ref{fig:Example2plot} for the three colors within the index regions $1\leq\ i\leq \tau_R$, $\tau_R<\ i\leq \tau_B$, and $\tau_B< i\leq (n-1)$, except that the \texttt{red} curve does not increase in the second index range. Furthermore, we observe that between the index ranges $0\leq i\leq 5$ the order $R_0<B_0<G_0$ is preserved, and then in $5 < i\leq (n-1)$ the order $R_i\leq G_i\leq B_i$ in the rainbow $c = (\texttt{blue, red, green})$ is achieved and preserved.
\end{example}
\section*{Acknowledgment}
Authors thank Yuzhou Gu for his suggestions to improve Theorem 2. This work has been supported in part by the German Federal Ministry of Education and Research (BMBF) under the Grant 16KIS1004 and the ARC Future Fellowship FT190100429.

\IEEEtriggeratref{8}
\bibliographystyle{IEEEtran}
\bibliography{DPreferencesISIT2022}

\end{document}